\documentclass[11pt,
onecolumn
]{article}

\usepackage[utf8]{inputenc}
\usepackage[T1]{fontenc}

\usepackage{amsmath,amssymb,amsfonts,mathtools}
\usepackage{bm}
\usepackage{physics}

\usepackage{graphicx}
\usepackage{booktabs}
\usepackage{multirow}

\usepackage{cite}
\usepackage{url}
\usepackage{xcolor}

\usepackage[
    colorlinks=true,
    linkcolor=black,
    citecolor=black,
    urlcolor=black
]{hyperref}

\usepackage[
    margin=0.75in
]{geometry}

\newtheorem{theorem}{Theorem}[section]
\newtheorem{definition}[theorem]{Definition}
\newtheorem{proposition}[theorem]{Proposition}

\newtheorem{proof}[theorem]{Proof}

\title{
Finite nonlinear mathematical structures induced by the Tsallis $q$-sum
}

\author{
    Ignacio S. Gomez$^{1,2}$
    \\[3pt]
    $^{1}$Departamento de Ciências Exatas e Naturais, Universidade Estadual do Sudoeste da Bahia, \\
			    BR 415, Itapetinga - BA, 45700-000, Brazil
    \\[3pt]
    $^{2}$PROFÍSICA – Programa de Pós-Graduação em Física, 
Universidade Estadual de Santa Cruz,\\ 45650-000 Ilhéus, BA, Brazil
    \\[3pt]
    \small
    \texttt{ignacio.gomez@uesb.edu.br}
}

\date{}

\begin{document}

\maketitle


\begin{abstract}
The Tsallis $q$-sum is one of the fundamental nonlinear composition laws of nonextensive statistical mechanics. Although it has been extensively investigated in continuous settings, its implications for finite mathematical structures remain less explored. Here we investigate this question by replacing ordinary additive laws with the Tsallis $q$-sum. We first introduce an axiomatic $q$-cardinality of finite sets satisfying a nonlinear additivity principle. Existence and uniqueness are established, leading to an explicit expression that continuously recovers the classical cardinality as $q\to1$. We then propose a nonlinear matrix composition induced by the same deformation and establish its principal algebraic properties, including associativity, the neutral element, and a characterization of its noncommutativity in terms of the ordinary matrix commutator. An illustrative example over $M_2(\mathbb Z_6)$ shows how the deformation can modify the center of a finite matrix algebra when $1-q$ is a zero divisor. These results indicate that the Tsallis $q$-sum provides a natural mechanism for constructing nonlinear finite mathematical structures and connects nonextensive statistical mechanics with finite algebra and nonlinear mathematical physics.
\end{abstract}


\section{Introduction}

Generalized composition laws play an important role in nonlinear science, mathematical physics, and statistical mechanics. A prominent example is the nonlinear addition introduced by Tsallis in nonextensive statistical mechanics \cite{Tsallis1988,CuradoTsallis1991}. For $x,y\in\mathbb R$, the corresponding $q$-sum is
\begin{equation}
 x\oplus_q y=x+y+(1-q)xy.
 \label{eq:qsum}
\end{equation}
It is one of the basic operations of the generalized $q$-algebra developed by Nivanen \emph{et al.} and Borges \cite{Nivanen2003,Borges2004}. Related deformations of algebraic structures, $q$-numbers, and generalized operations have subsequently been studied in several directions \cite{Cardoso2008,Lobao2009,ElKaabouchi2009}. These developments are part of a broader program involving generalized thermostatistics, generalized central-limit behavior, and information-theoretic formulations \cite{Tsallis2009,Umarov2008,GellMannTsallis2004,Suyari2006}. Deformed exponential and logarithmic structures provide a broader mathematical setting for generalized thermostatistics \cite{Naudts2002,Kaniadakis2002,Naudts2011}. Formal-group and measure-theoretical approaches further connect composability and invariance translation with generalized entropies and algebraic structures \cite{Tempesta2011,Tempesta2016A,Tempesta2016B,EncisoTempesta2017,Gomez2018}.

Most applications of these constructions concern generalized statistical mechanics, deformed calculus, or continuous analytical structures. A natural complementary question is whether the same composition law can systematically generate new finite mathematical structures. Finite sets and finite matrix algebras provide particularly simple settings in which such a deformation can be formulated explicitly and tested algebraically.

The purpose of this work is to develop this viewpoint. We first replace the ordinary additivity of finite cardinalities by a $q$-additivity axiom. We prove existence and uniqueness and derive the corresponding closed expression, whose classical limit is the ordinary cardinality. The same deformation principle is then applied to finite matrix algebras. The resulting nonlinear matrix composition is associative, possesses the zero matrix as identity, and has a noncommutativity controlled exactly by the ordinary matrix commutator.

The construction is related to, but distinct from, earlier studies of deformed numbers and generalized algebraic operations \cite{Cardoso2008,Lobao2009,Tsallis1994,BorgesKodamaTsallis2022}. In particular, the present approach starts from an axiomatic finite-counting problem and then transfers the same composition law to matrix addition. This provides a common framework in which finite counting and matrix composition appear as two realizations of a single deformation principle.

The paper is organized as follows. Section~2 reviews the $q$-sum and motivates its application to finite counting. Section~3 develops the axiomatic $q$-cardinality and proves existence, uniqueness, and the closed formula. Section~4 introduces the nonlinear matrix composition and gives its basic algebraic properties together with an illustrative example over $M_2(\mathbb Z_6)$. Section~5 discusses the mathematical and physical interpretation of the construction. Section~6 presents the conclusions.

\section{Preliminaries}

The $q$-sum in Eq.~(\ref{eq:qsum}) satisfies
\begin{equation}
\lim_{q\to1}(x\oplus_q y)=x+y,
\end{equation}
and is commutative and associative on $\mathbb R$, with neutral element zero. It is therefore a commutative monoid operation \cite{Borges2004,Nivanen2003}. The associated $q$-algebra, together with generalized exponential and logarithmic functions, forms a standard mathematical framework for nonextensive thermostatistics \cite{Borges2004,Naudts2011}. The algebraic properties of related deformed operations have also been investigated explicitly \cite{Cardoso2008,Lobao2009,ElKaabouchi2009}.

For disjoint finite sets $A$ and $B$, ordinary cardinality obeys
\begin{equation}
 |A\cup B|=|A|+|B|.
 \label{eq:classical}
\end{equation}
We shall replace Eq.~(\ref{eq:classical}) by
\begin{equation}
 |A\cup B|_q=|A|_q\oplus_q|B|_q,
 \label{eq:axiom}
\end{equation}
while retaining the usual normalization of the empty set and singleton. This is the finite-counting deformation investigated in the next section.

\section{An Axiomatic Construction of the $q$-Cardinality}

The ordinary cardinality of finite sets is characterized by additivity. We replace this property by Eq.~(\ref{eq:axiom}) while preserving the basic normalization conditions.

\begin{definition}
A \emph{$q$-cardinality} is a mapping $|\cdot|_q:\mathcal F\to\mathbb R$, where $\mathcal F$ is the family of finite sets, satisfying
\begin{enumerate}
\item[(A1)] $|\varnothing|_q=0$;
\item[(A2)] $|\{x\}|_q=1$ for every singleton;
\item[(A3)] $|A\cup B|_q=|A|_q\oplus_q|B|_q$ for disjoint finite sets.
\end{enumerate}
\end{definition}

Let $a_n$ denote the $q$-cardinality of any finite set with $n$ elements. Axiom (A3) gives
\begin{equation}
 a_{n+1}=a_n\oplus_q1=(2-q)a_n+1,
 \qquad a_0=0.
 \label{eq:recurrence}
\end{equation}

\begin{theorem}[Existence]
The recurrence in Eq.~(\ref{eq:recurrence}) defines a $q$-cardinality satisfying the three axioms above.
\end{theorem}
\begin{proof}
For every $n\ge0$, Eq.~(\ref{eq:recurrence}) recursively defines a unique value $a_n$. If $A$ and $B$ are disjoint, their union has $|A|+|B|$ elements, and repeated use of the recurrence gives Eq.~(\ref{eq:axiom}). Hence the resulting map satisfies (A1)--(A3).
\end{proof}

\begin{theorem}[Uniqueness]
The $q$-cardinality satisfying Definition~1 is unique.
\end{theorem}
\begin{proof}
Any two $q$-cardinalities satisfying (A1)--(A3) obey the same recurrence (\ref{eq:recurrence}) with the same initial condition $a_0=0$. Induction on $n$ therefore gives identical values for every finite set.
\end{proof}

Solving Eq.~(\ref{eq:recurrence}) yields
\begin{equation}
 |A|_q=\frac{(2-q)^{|A|}-1}{1-q},
 \qquad q\ne1.
 \label{eq:qcardinality}
\end{equation}
Thus
\begin{equation}
 \lim_{q\to1}|A|_q=|A|.
 \label{eq:qcardlimit}
\end{equation}
Introducing $Q=2-q$, Eq.~(\ref{eq:qcardinality}) becomes
\begin{equation}
 |A|_q=\frac{Q^{|A|}-1}{Q-1},
\end{equation}
which is the standard $Q$-integer of quantum calculus \cite{KacCheung2002}. The connection between $q$-operations and deformed numbers has also been studied in \cite{Lobao2009,BorgesKodamaTsallis2022}. Here, however, the expression is obtained directly from finite-set axioms rather than postulated as a deformed number.

\section{A Nonlinear Matrix Composition Induced by the Tsallis $q$-Sum}

The same deformation principle can be transferred from finite counting to matrix addition. Let $M_n(R)$ denote the algebra of $n\times n$ matrices over a commutative ring $R$, with $q\in R$. We define
\begin{equation}
 A\boxplus_qB=A+B+(1-q)AB.
 \label{eq:matrixcomposition}
\end{equation}
For $q\to1$, $A\boxplus_qB\to A+B$. The construction is motivated by the scalar $q$-sum, but matrix multiplication introduces a new feature: the composition can become noncommutative.

\subsection{Elementary properties}

\begin{proposition}
The operation $\boxplus_q$ is associative.
\end{proposition}
\begin{proof}
Direct expansion gives
\begin{align*}
(A\boxplus_qB)\boxplus_qC
&=A+B+C+(1-q)(AB+AC+BC)\\
&\quad +(1-q)^2ABC,
\end{align*}
which is identical to $A\boxplus_q(B\boxplus_qC)$.
\end{proof}

\begin{proposition}
The zero matrix is the neutral element of $\boxplus_q$.
\end{proposition}
\begin{proof}
Since $A0=0=0A$, one has $A\boxplus_q0=0\boxplus_qA=A$.
\end{proof}

\begin{proposition}
The noncommutativity of $\boxplus_q$ is governed by the ordinary matrix commutator.
\end{proposition}
\begin{proof}
Subtracting the two possible orderings gives
\begin{equation}
 A\boxplus_qB-B\boxplus_qA=(1-q)(AB-BA)=(1-q)[A,B].
 \label{eq:commutator}
\end{equation}
\end{proof}

Thus commuting matrices remain commuting under $\boxplus_q$, whereas noncommuting matrices acquire a deformation proportional to $(1-q)$. In particular, the noncommutative contribution vanishes continuously in the classical limit $q\to1$.

The composition is also covariant under similarity transformations. For every invertible $P$,
\begin{equation}
 P^{-1}(A\boxplus_qB)P=(P^{-1}AP)\boxplus_q(P^{-1}BP),
\end{equation}
which follows immediately from Eq.~(\ref{eq:matrixcomposition}). Hence the deformation is compatible with the standard change of matrix representation.

\subsection{Illustrative example: $M_2(\mathbb Z_6)$}

Consider the finite matrix algebra $M_2(\mathbb Z_6)$. Its ordinary multiplicative center is
\begin{equation}
 Z(M_2(\mathbb Z_6))=\{\lambda I:\lambda\in\mathbb Z_6\}.
\end{equation}
For the deformed composition, define
\begin{equation}
 Z_q(M_2(\mathbb Z_6))
 =\{A:A\boxplus_qB=B\boxplus_qA,\ \forall B\}.
 \label{eq:qcenter}
\end{equation}
Equation~(\ref{eq:commutator}) gives the condition
\begin{equation}
 (1-q)[A,B]=0,\qquad \forall B.
 \label{eq:qcentercondition}
\end{equation}
Therefore
\begin{equation}
 Z(M_2(\mathbb Z_6))\subseteq Z_q(M_2(\mathbb Z_6)),
\end{equation}
while equality is obtained when $1-q$ is a unit in $\mathbb Z_6$. If $1-q$ is a zero divisor modulo $6$, additional matrices can satisfy Eq.~(\ref{eq:qcentercondition}). Thus the deformed center depends not only on the underlying finite matrix algebra but also on the arithmetic of the deformation parameter.

\section{Discussion}

The two constructions developed above originate from the same replacement of ordinary addition by the Tsallis $q$-sum. The first acts on finite counting, whereas the second acts on matrix composition. In both cases, the ordinary theory is recovered continuously as $q\to1$.

The $q$-cardinality is uniquely fixed by a small set of axioms, and its closed form is a $Q$-integer with $Q=2-q$. This places the construction in direct contact with the broader literature on deformed numbers and $q$-calculus \cite{KacCheung2002,Lobao2009}. At the same time, the matrix construction extends the same composition principle to a setting in which multiplication is generally noncommutative. Equation~(\ref{eq:commutator}) is therefore central: it separates the nonlinear deformation parameter from the intrinsic noncommutativity of the matrix algebra.

A useful physical interpretation is obtained by viewing $q=1$ as the undeformed or classical additive regime. Away from this limit, the composition law contains a nonlinear correction, while the matrix case can additionally display noncommutativity. This should be understood as an algebraic analogy rather than as a claim that the present construction by itself constitutes a classical--quantum correspondence. The same viewpoint is consistent with the broader use of deformation and composability in generalized statistical mechanics and formal-group approaches \cite{Tempesta2011,Tempesta2016A,Tempesta2016B,EncisoTempesta2017,Naudts2011}.

Table~\ref{tab:comparison} summarizes the main differences between the ordinary and $q$-deformed structures. It highlights that finite counting and matrix composition are governed by one common deformation rule while exhibiting different algebraic consequences.

\begin{table}[t]
\centering
\small
\begin{tabular}{lll}
\hline
Feature & Classical & $q$-deformed\\
\hline
Composition & $x+y$ & $x\oplus_q y$\\
Finite counting & $|A|$ & $|A|_q$\\
Disjoint union & Additive & $q$-additive\\
Matrix law & $A+B$ & $A\boxplus_qB$\\
Commutativity & Additive & $(1-q)[A,B]$\\
Counting formula & $n$ & $\dfrac{(2-q)^n-1}{1-q}$\\
Classical limit & -- & $q\to1$\\
\hline
\end{tabular}
\caption{Comparison of the classical and $q$-deformed finite structures introduced in this work.}
\label{tab:comparison}
\end{table}

The construction suggests several extensions, including deformed finite groups, $q$-Lie-type structures, generalized combinatorial invariants, and applications to nonlinear operator and information-theoretic settings. These possibilities are natural directions for future work, but are not assumed in the present results.

\section{Conclusions}

We have investigated a common deformation principle based on the Tsallis $q$-sum and applied it to two finite mathematical structures. First, an axiomatic $q$-cardinality was constructed for finite sets. The axioms lead to a unique counting function with the explicit form in Eq.~(\ref{eq:qcardinality}), and the ordinary cardinality is recovered in the limit $q\to1$.

Second, we introduced the nonlinear matrix composition $A\boxplus_qB=A+B+(1-q)AB$. We proved associativity, identified the zero matrix as its neutral element, and established the exact relation between its noncommutativity and the ordinary commutator through Eq.~(\ref{eq:commutator}). The example of $M_2(\mathbb Z_6)$ further shows that the corresponding $q$-center can depend on the arithmetic properties of $1-q$, providing a concrete finite realization of the deformation.

Table~\ref{tab:comparison} emphasizes the common origin of the two constructions and the recovery of the classical additive theory at $q=1$. Overall, the results support the interpretation of the Tsallis $q$-sum as a useful mechanism for constructing and analyzing nonlinear finite mathematical structures. Future work may investigate the associated group-theoretical, Lie-algebraic, combinatorial, and dynamical extensions.

\section*{Acknowledgements}

Ignacio S. Gomez acknowledges support from the Department of Exact and Natural Sciences of the State University of Southwest Bahia (UESB), Itapetinga, Bahia, Brazil; from PROFÍSICA (UESC), Ilhéus, Bahia, Brazil; and from the Conselho Nacional de Desenvolvimento Científico e Tecnológico (CNPq), Grant No. 316131/2023-7.

\bibliographystyle{unsrt}
\bibliography{q-references}

\end{document}